\documentclass[letterpaper, 10 pt, conference]{ieeeconf}  

\IEEEoverridecommandlockouts                              
\usepackage{graphicx}
\usepackage{amsmath}
\usepackage{ntheorem}
\usepackage{amssymb}
\usepackage{tabularx}
\usepackage{mathtools}
\usepackage{amsfonts}
\usepackage{color}
\usepackage{subfigure}
\usepackage{algorithm}
\usepackage{uri}
\usepackage{textcomp}
\usepackage{siunitx}

\definecolor{aoenglish}{rgb}{0.0, 0.5, 0.0}
\definecolor{darkblue}{rgb}{0.0, 0.0, 0.55}
\definecolor{darkmagenta}{rgb}{0.55, 0.0, 0.55}
\definecolor{electricviolet}{rgb}{0.56, 0.0, 1.0}
\definecolor{electricyellow}{rgb}{1.0, 1.0, 0.0}
\definecolor{forestgreen}{rgb}{0.13, 0.55, 0.13}
\definecolor{fuchsia}{rgb}{1.0, 0.0, 1.0}
\definecolor{gamboge}{rgb}{0.89, 0.61, 0.06}
\definecolor{goldenpoppy}{rgb}{0.99, 0.76, 0.0}
\definecolor{indigo}{rgb}{0.29, 0.0, 0.51}
\definecolor{internationalorange}{rgb}{1.0, 0.31, 0.0}
\definecolor{lava}{rgb}{0.81, 0.06, 0.13}
\definecolor{selectiveyellow}{rgb}{1.0, 0.73, 0.0}
\definecolor{turquoiseblue}{RGB}{0,191,255}
\definecolor{darkgreen}{RGB}{0, 100, 0}

\usepackage{xcolor}

\renewenvironment{proof}{\paragraph*{Proof.}}{\hfill$\square$}

\newtheorem{assumption}{Assumption}
\newtheorem{lemma}{Lemma}

\newtheorem{proposition}{Proposition}

\newtheorem{remark}{Remark}

\let\labelindent\relax
\usepackage{enumitem}
\renewcommand{\theenumi}{\arabic{enumi}}

\renewcommand{\theenumii}{\arabic{enumii}}

\renewcommand{\theenumiii}{\arabic{enumiii}}

\renewcommand{\theenumiv}{\arabic{enumiv}}

\newcommand{\E}{{\mathbb E}}

\newcommand{\Var}{\operatorname{Var}}

\usepackage{algpseudocode}
\usepackage{caption}
\usepackage{cite}

\title{\LARGE \bf
Forgetting While Remembering, \\ an Invariant Online Data-Driven Predictive Control Formulation
}

\author{Alessandro Chiuso\thanks{A. Chiuso is with the Department of Information Engineering, University of Padova, Italy. Email address: {\tt \small alessandro.chiuso@unipd.it}}, 
Florian D\"{o}rfler\thanks{F. Dörfler is with the Automatic Control Laboratory, Swiss Federal Institute of Technology (ETH), Switzerland. Email address: {\tt\small dorfler@control.ee.ethz.ch}.}, and 
Keith Moffat\thanks{K. Moffat is with the Department of Electrical and Electronic Engineering, University of Melbourne, Australia. Email address: {\tt\small keith.moffat@unimelb.edu.au}}
}

\begin{document}

\maketitle
\thispagestyle{empty}
\pagestyle{empty}


\begin{abstract}
Low signal-to-noise ratio (SNR) data is a core challenge of online Data-Driven Predictive Control (DPC) for linear, time-varying systems. 
This paper proposes a Bayesian, online DPC framework based on autoregressive models with exogenous inputs (ARX) that uses an externally-provided prior, which encodes inductive bias such as smooth system dynamics and stability, to safeguard performance when SNR is low. 
The posterior estimate of the ARX parameter is propagated forward in time using a Kalman filter with a state equation defined by an adaptation-rate hyperparameter, which is adjusted online to track the rate at which the underlying system dynamics evolve.
The Kalman filter's posterior mean and covariance determine the DPC's Final Control Error cost function, which is the posterior expectation of the quadratic cost function.
Critically, the Kalman filter's process equation is chosen so that, regardless of the hyperparameter adaptation, the prior distribution is invariant over time.
Thus, while old data is forgotten, the 
prior is not.
DPC tracking experiments on a time-varying second-order system demonstrate the efficacy of the proposed method. 
\end{abstract}
\section{INTRODUCTION}

Data-Driven Predictive Control (DPC) has gathered attention in recent years due to its ease of implementation 
\cite{Coulson:2019, Dorfler:2023, chiuso2025harnessing, Breschi:2023, Berberich:2020, Lazar:2024, Verheijen:2023, MoffatCDC2024}.
Much of the literature has focused on the linear, time-invariant setting, in which the system's behavior can be modeled as a linear combination of sufficiently-exciting data \cite{Willems2005}.
In this paper, we focus on adapting the controller in the linear, time-varying setting.

Like other Adaptive Control formulations, Online DPC is challenging because the appropriate adaptation rate is system-dependent and time-varying.
When the system changes quickly, relying too heavily on old data can produce an incorrect system representation and degrade control performance. 
Conversely, if the system changes slowly, relying too heavily on recent data can produce online training data SNR deficiencies, loss of stability, and robustness problems such as bursting phenomena \cite{ANDERSON1985}.

Focusing on adaptive predictive control approaches, \cite{Wang2025} introduces a change detection mechanism which triggers the collection and use of fresh data, based on which a controller is re-designed. 
A recursive Data-Enabled Predictive Control (DeePC, \cite{Coulson:2019}) is proposed in \cite{Shi2024}. 
Similarly, \cite{Berberich:2022} uses the recent input-output data to form a local-linear approximation of the system dynamics. 
The approach in \cite{Berberich:2022} is subspace-based and can therefore produce biased predictions when applied to noisy closed-loop data in the online/adaptive context \cite{ljung1996subspace, MoffatECC2024}. 
In contrast, \cite{Shi2024} addresses this issue through a two-level optimization problem that constructs multi-step predictions from a single-step predictor.
However, neither method uses an adaptive forgetting factor, and thus both are susceptible to low-SNR data issues.

Bayesian recursive-in-time parameter estimation, including the Kalman filter, provides an alternative framework for model adaptation. 
It has a long history, including, for example, \cite{sarris1973bayesian,MehraTAC1970,MehraTAC1972}, and we refer the reader to \cite{BarShalomIEEEAcces2020} for a recent overview of the literature.
Recursive Least Squares can be viewed as a special case of the Kalman filter, where the forgetting factor is  replaced by the  process noise variance \cite{sayed1994state, lai2024adaptive}.
The literature contains many examples of adjusting the forgetting rate/process noise online, based on observed data \cite{MehraTAC1970, Fortescue1981VariableForgetting, Bruce2020VariableRateForgetting}, including methods that adjust the forgetting factor online by marginal-likelihood optimization \cite{prando2016online}. 

Variable forgetting rates \cite{MehraTAC1970, Fortescue1981VariableForgetting, Bruce2020VariableRateForgetting, prando2016online} can help address low SNR by reducing the forgetting rate when the dynamics vary slowly.
Another way to improve estimation performance with low SNR data is to use a prior. 
Encoding inductive bias (e.g., the system is smooth and/or stable) as a prior anchors the Bayesian estimate, reducing the impact of measurement variance.

We propose an adaptive Kalman filter ARX parameter estimation algorithm with the following characteristics.
\textit{First}, rather than adapting the process noise (equivalent to the forgetting factor in the Recursive Least Squares context), as is common in the literature \cite{MehraTAC1970, Fortescue1981VariableForgetting, Bruce2020VariableRateForgetting, prando2016online}, 
we \textit{adapt the process equation online to ``forget'' older data.} 
\textit{Second}, the adaptive Kalman filter is designed so that, regardless of how the process equation is adapted, \textit{the prior distribution of the model parameter is invariant}. 
To the best of our knowledge, the proposed Kalman filter parameter estimator with an invariant prior and adaptive process equation is novel in the context of adaptive filtering and control.


In addition to the parameter estimator, we also employ a Bayesian DPC optimization. 
We build off the work in \cite{chiuso2025harnessing}, which introduced the Final Control Error (FCE) for an uninformative prior and showed that the mean and variance of the ARX parameters are sufficient for quadratic FCE-minimizing predictive control.
Conveniently, the Kalman filter updates  the posterior mean and variance of the ARX parameters, which can be directly plugged into the FCE-minimizing DPC optimization.
As this formulation is based on a single-step-predictor, it does not suffer from closed-loop bias issues  \cite{ljung1996subspace, CHIUSOAuto2007, MoffatECC2024}.

Summarizing, this paper proposes an online DPC algorithm for linear, time-varying systems which 
\begin{enumerate}
    \item recursively updates  the predictor mean and variance using a Kalman filter, 
    \item updates the Kalman filter process equation using Empirical Bayes-based gradient updates, and 
    \item parametrizes the predictive control cost function to minimize the FCE (the posterior expectation of the quadratic predictive control cost).
\end{enumerate}

\section{PROBLEM STATEMENT}

We consider the problem of controlling a linear, time-varying, and stochastic system 
with input $u_t \in \mathbb{R}$ and output $y_t \in \mathbb{R}$, $t\in {\mathbb Z}$.
We present the SISO case for clarity, recognizing that the construction extends to MIMO systems.

We define the joint input-output variable
\begin{equation*}
z_t:=\begin{bmatrix}
y_t\\
u_t
\end{bmatrix} \in \mathbb{R}^{2},
\end{equation*}
and focus on a receding horizon control problem of the following form. 
Given a control horizon $T$, define the future input and output sequences
\begin{equation*}
\begin{array}{rcl}
u_f&:=&\begin{bmatrix}
u_t^{\top} &
u_{t+1}^{\top} &
\cdots &
u_{t+T-1}^{\top}
\end{bmatrix}^{\top},
\\
y_{f}&:=&\begin{bmatrix}
y_t^{\top} & y_{t+1}^{\top} & \cdots & y_{t+T-1}^{\top}
\end{bmatrix}^{\top},
\end{array}
\end{equation*}
where $y_r\in\mathbb{R}^T$ is the desired output.
We are interested in minimizing the ground-truth cost
\begin{equation}\label{eq:cost}
\mathcal{J}(u_f,y_f):= \|y_f - y_r\|^2 +  \frac{1}{q}\| u_f\|^2, \quad q>0.
\end{equation}
In the data-driven setting, a model for the system is not available and, at any time $t$, past data 
$${\cal D}^{-}_t:=\{z_\tau\}_{\tau = 1,..,t-1}$$ 
is used to predict $y_f$ as a function of $u_f$.

In contrast to a certainty-equivalent cost, 
we occupy ourselves with the conditional expectation of the cost \eqref{eq:cost}, 
\begin{equation}\label{eq:FCE}
F_t(u_f): = {\mathbb E}[\mathcal{J}(u_f,y_f)|{\cal D}^{-}_t].
\end{equation}
This conditional-on-data expected cost was named the FCE in \cite{chiuso2025harnessing}, which investigated the offline problem.
Conditioning on data ${\cal D}^{-}_t$, \eqref{eq:FCE} implies that $y_f$ is distributed according to the dynamics \eqref{eq:predictorTV_ARX} that ${\cal D}^{-}_t$ has revealed. 
For the prior, we consider the class of linear, possibly time-varying systems described with a one-step-ahead predictor, 
\begin{align}\label{eq:predictorTV_ARX}
 y_t &= \sum_{k=1}^{\hat\rho} \phi_{t,k} z_{t-k} + e_t,
\end{align}
where $e_t$ is the innovation and
the predictor-length $\hat\rho$ is chosen using the Akaike Information Criterion (AIC) \cite{Akaike1974}, as discussed in \cite{chiuso2025harnessing}, and, for the SISO setting, $\phi_{t,k}\in\mathbb{R}^{1\times 2}$.
We assume, for simplicity, that  $\hat\rho$ is constant over time.
Regarding the prediction parameter $\phi_{t,k}$, the subscript $k$ indicates the number of time steps into the past, while the subscript $t$ indicates that the parameter is time-varying.

Theorem 1 in \cite{chiuso2025harnessing} demonstrated that, for linear, time-invariant systems, the FCE admits an asymptotic decomposition into a certainty-equivalent cost and a quadratic regularization term accounting for the uncertainty-induced cost.
For the time-varying case, we use the corresponding approximation
\begin{equation}\label{eq:sepPrinc}
F_t(u_f)
=
\mathbb{E}[\mathcal{J}(u_f,y_f)\mid\mathcal{D}_t^-]
\approx
J_t(u_f)+r_t(u_f),
\end{equation}
where $J_t(u_f)$, defined later in \eqref{eqn:JtDef}, is the certainty-equivalent cost and $r_t(u_f)$ is the regularization term.


We define the ARX parameter vector 
\begin{align*}
    \theta_t := [\phi_{t,\hat\rho} \quad \cdots \quad \phi_{t,1} ]^\top,
\end{align*}
\begin{assumption}
    The externally-provided prior $\theta_t \sim \mathcal{N}(0,Q)$. 
\end{assumption}
The Gaussian prior ensures that the Kalman Filter produces the exact posterior mean and covariance.
The zero-mean assumption is made for notational simplicity; the method extends directly to a nonzero prior mean.

\noindent \textbf{Problem Formulation}:
At each moment $t$, given data ${\cal D}^{-}_t$ and the prior variance  on the predictor coefficients $Q$, determine the control action $u_f$ that minimizes the FCE $F_t(u_f)$.


\section{Forgetting While Remembering: \\
An Invariant Adaptation Model}

Before describing the adaptive Kalman filter formulation, we consider two extreme cases of time-varying systems: 
\begin{enumerate}
    \item \textit{A white noise system}: $\theta_t$ is temporally uncorrelated such that $\theta_t$ provides no information on $\theta_{t+1}$. 
    \item \textit{A time-invariant system}: $\theta_t$ is constant in time and thus $\theta_t = \theta_{t+1}$. 
\end{enumerate}
Each extreme case has modeling benefits and weaknesses---Estimating a parameter using the white-noise system model allows quick adaptation when the underlying system changes, but only uses the most recent measurement.
Estimating a parameter using the time-invariant assumption, on the other hand, weighs all data equally and thus does not adapt effectively when the underlying system changes.
We propose a method that interpolates between these two extreme cases in an adaptive, data-driven manner in order to adjust to time-varying evolution rates.


\subsection{System Prior}

We assert the following process equation for the ARX parameter:
\begin{equation}
\label{eq:ssprior}
\begin{array}{rcl}
\theta_{t+1} & = & \alpha \theta_t + \left(\sqrt{1-\alpha^2}\right) w_t, \quad \alpha \in [0,1],
\end{array}
\end{equation}
where $\alpha$ is the parameter that determines the process equation, and $w_t$ is zero-mean white noise with variance $Q$. 

All process equations assert a prior on $\theta_{t}$. The commonly-used random walk (\cite{MehraTAC1970, Fortescue1981VariableForgetting, Bruce2020VariableRateForgetting}) asserts a prior that is forgotten over time: 
$\mathrm{Var}(\theta_t) \rightarrow \infty$ as $t \rightarrow \infty$. 
On the other hand, as stated in the following proposition,
\eqref{eq:ssprior} asserts a zero mean prior that is invariant (not forgotten over time). 
The prior's variance, $\mathrm{Var}(\theta_t)$, encodes the correlation between parameters and can be used to assert properties such as smoothness or stability.




\begin{proposition}\label{prop:invarance}
If $\theta_0\sim \mathcal{N}(0,Q)$, $w_t \sim \mathcal{N}(0,Q)$ is Gaussian white noise independent of $\theta_0$, then, given \eqref{eq:ssprior}, $\theta_t \sim \mathcal{N}(0,Q)$ for all \(t\).
\end{proposition}

\begin{proof}
$\theta_{t+1}$ is an affine function of $\theta_t$ and $w_t$.
Since $\theta_0$ and $w_t$ are Gaussian, it follows by induction that $\theta_t$ is Gaussian for all $t$.

Mean: 
$\mathbb{E}[\theta_{t+1}] = \alpha \mathbb{E}[\theta_t]$.
As $\mathbb{E}[\theta_0] = 0$, $\mathbb{E}[\theta_t] = 0$.

Variance:
$
    \Var\{\theta_{t+1}\} = \alpha^2 \Var\{\theta_{t}\} + (1-\alpha^2) \Var\{w_t\}.
$
The result follows from $\Var\{\theta_{0}\} = \Var\{w_t\} = Q$.
\end{proof}

Thus, by modeling the evolution of the parameter $\theta_t$ with \eqref{eq:ssprior} and setting $\Var\{\theta_0\} = Q$, and  $\Var\{w_t\} = Q$, the externally-provided prior variance $Q$ applies at all times $t$. 

\begin{remark}
    The parameter $\alpha \in [0,1]$ models how fast  $\theta_t$ varies. 
    Large $\alpha$ implies that the system is slowly-time-varying with the time-invariant $\theta_{t+1} = \theta_t$ case at the $\alpha=1$ limit.
    Small $\alpha$ implies that the system is rapidly-time-varying with the white noise $\E[\theta_{s} \theta_t^\top ] = Q \delta_{t-s}$ case at the $\alpha=0$ limit.
\end{remark}

\noindent In Section \ref{sec:Adapt:hyperparam} we describe how the parameter $\alpha$ is estimated on-line, thus the adaptation speed is learned from data. 

\subsection{Recursive, Bayesian ARX parameter estimation}

The parameter model \eqref{eq:ssprior}, paired with the predictor equation \eqref{eq:predictorTV_ARX}, yields a (time varying) linear state space model of the form
\begin{equation}
\label{eq:ssKalman}
\begin{array}{rcl}
\theta_{t+1} & = & \alpha \theta_t + \left(\sqrt{1-\alpha^2}\right) w_t, \\
y_t & = & C_t \theta_t + e_t,
\end{array}
\end{equation}
where $C_t := \begin{bmatrix}
z_{t-\hat\rho}^\top & \cdots & z_{t-1}^\top
\end{bmatrix}$ is the matrix formed from the most recent input-output data at time $t$ such that 
$$
C_t \theta_t = \sum_{k=1}^{\hat\rho} \phi_{t,k} z_{t-k},
$$ 
and $e_t \sim \mathcal{N}(0,\sigma^2)$ is white Gaussian noise, independent of $\theta_t$ and the process noise $w_t$, which models the unavoidable one-step-ahead prediction error $y_t-C_t\theta_t$.

Interpreted as a Kalman filter, $e_t$ is the ``measurement noise'' and $w_t$ is the ``process noise.''
The conditional mean and variance of $\theta_t$ can be estimated recursively with the Kalman filter as follows, 
\begin{equation}\label{eq:Kalman:update}
\begin{array}{rcl}
\Lambda_t & = & C_t \Sigma_{t|t-1} C_t^\top + \sigma^2, \\
K_t  & = & \Sigma_{t|t-1} C_t^\top \Lambda^{-1}_t, \\
\hat \theta_{t|t} &=& \hat \theta_{t|t-1} + K_{t}\left(y_t - C_t\hat\theta_{t|t-1}\right), \\
\Sigma_{t|t} & = & \Sigma_{t|t-1} - K_t \Lambda_t K_t^\top, \\
\hat \theta_{t+1|t} &=& \alpha \hat \theta_{t|t}, \\
\Sigma_{t+1|t} &=& \alpha^2\Sigma_{t|t}
+ (1-\alpha^2) Q ,
\end{array}
\end{equation}
where $\hat \theta_{t|t-1}$ and $\hat\theta_{t|t}$ denote the predicted and filtered estimates of $\theta_t$ using measurements up to time $t-1$ and $t$, respectively, and  
$\Sigma_{t|t-1}$, $\Sigma_{t|t}$ are their respective covariances. 
Summarizing the relevant variances:
\begin{flushleft}
\begin{tabularx}{\columnwidth}{@{}p{0.34\columnwidth}>{\raggedright\arraybackslash}X@{}}
$Q = \operatorname{Var}(w_t)$
& Covariance of the process noise. Also the variance of $\theta_t$ prior to any measurements (see Prop. \ref{prop:invarance}). \\

$\Sigma_{t|t-1} \!\!=\!\! \operatorname{Var}(\theta_t \!\!\mid\!\! {\cal D}_t^-)$
& 
Covariance of $\theta_t$ before seeing $y_t$. \\

$\Sigma_{t|t} = \operatorname{Var}(\theta_t \!\mid\! {\cal D}_t^-,y_t)$
& 
Covariance of $\theta_t$ after seeing $y_t$. \\

$\sigma^2 = \operatorname{Var}(e_t)$
& Prediction error variance (oracle/unavoidable error variance). \\

$\Lambda_t = \operatorname{Var}(\epsilon_t)$ 
& Innovation variance 
(non-oracle prediction error variance).
\end{tabularx}
\end{flushleft}
We define the Kalman filter innovation, emphasizing its dependence on $\alpha$, as
\begin{align*}
\epsilon_t(\alpha) := y_t - C_t \alpha \hat\theta_{t-1|t-1},
\end{align*}
which differs from the oracle prediction error $e_t$, as $\epsilon_t$ depends on $\hat\theta_{t|t-1}$, rather than the true $\theta_t$.
The filter innovation variance is 
\begin{align*}
\Lambda_t(\alpha) = C_t\!\left(
\alpha^2\Sigma_{t-1|t-1}+(1-\alpha^2)Q
\right)\!C_t^\top+\sigma^2.
\end{align*}
In this paper, to simplify notation, we restrict the class of systems to SISO systems, thus $\sigma^2$ and $\Lambda_t$ are scalars, whereas $Q$, $\Sigma_{t|t-1}$, and $\Sigma_{t|t}$ are $2\hat{\rho}\times 2\hat{\rho}$ matrices. 

\begin{remark}
    The Kalman filter \eqref{eq:Kalman:update} propagates the posterior belief forward in time. 
    Because the  prior is zero-mean, the process equation pulls $\hat \theta$ towards zero at each time step: $\hat{\theta}_{t|t-1} = \alpha \hat{\theta}_{t-1|t-1}$, $\alpha \in [0,1]$. 
    This pull is counteracted by the Kalman filter's measurement update: $\hat{\theta}_{t|t} = \hat{\theta}_{t|t-1} + K_t\bigl(y_t - C_t \hat{\theta}_{t|t-1}\bigr)$, which pushes the ARX parameter estimate to describe the observed data $C_t$. 
\end{remark}

\subsection{Recursive Final Control Error}

We build a $T$-step predictor, as in \cite{chiuso2025harnessing}, by recursively applying the one-step ARX predictor with coefficients $\hat{\theta}_{t|t-1}$. At each step, the candidate input and predicted output are appended to the $\hat{\rho}$-long data window, and the oldest pair is discarded (see also the ``Fixed-Length Predictor'' in \cite{MoffatCDC2024}).
We denote the resulting $u_f$-dependent prediction mean by $\hat{y}_{f|t-1,u_f}$.
Following \cite{chiuso2025harnessing}, the $u_f$-dependent component of the covariance of $y_f$, conditioned on $\mathcal{D}_t^-$ and $u_f$, is
\begin{align*}
\Sigma_{y_f|t-1,u_f}
&=
M_t(u_f)\Sigma_{t|t-1}M_t(u_f)^\top,
\end{align*}
where $M_t(u_f)$ is affine in $u_f$.
This term therefore captures how uncertainty in the ARX coefficients propagates into the multi-step output prediction for a candidate $u_f$.

Using this prediction mean and covariance, the FCE \eqref{eq:sepPrinc} admits an approximate decomposition into the certainty-equivalent term
\begin{align}\label{eqn:JtDef}
J_t(u_{f}) &:= \|\hat y_{f|t-1,u_f} - y_r\|^2+\frac{1}{q}\|u_f\|^2, 
\end{align}
and a variance term 
\begin{align*}
r_t(u_f) &:= \operatorname{Tr}\!\big(\Sigma_{y_f|t-1,u_f}\big).
\end{align*}
The control sequence is therefore determined by
\begin{equation}\label{eq:FCEcostSISO}
u_f^\star \in \arg\min_{u_f}\; J_t(u_f)+r_t(u_f),
\end{equation}
which is an unconstrained, strictly convex quadratic program. Its unique solution is obtained by solving the linear system arising from the first-order optimality condition.

\section{Learning to forget: \\self-tuning of the adaptation speed}\label{sec:Adapt:hyperparam}

Following the empirical Bayes approach \cite{Maritz:1989}, the hyperparameter $\alpha$ that governs the evolution of the parameter $\theta$ can be estimated by maximizing the marginal likelihood, which is the likelihood of observed data once the dependence on $\theta_t$ has been marginalized to leave a Gaussian for the one-step predictive distribution of $y_t$. 
This marginalization is implicit in the Kalman Filter.

Conditioning on the data and previous hyperparameter estimates, the conditional negative log-likelihood of $y_t$ is
\begin{align*}
L_t(\alpha)
:=
-2 \log p_{\alpha}\!\left(
y_t \mid {\cal D}_t^-,\{\alpha_\tau\}_{\tau=1,\ldots,t-1}
\right).
\end{align*}
In our online setting, after $y_t$ becomes available, $\alpha_t$ is obtained by applying a Newton/gradient step to $L_t(\alpha)$. 
The value $\alpha_t$ is then used in the prediction step of the Kalman Filter \eqref{eq:Kalman:update} from time $t$ to $t+1$ to produce $\hat{\theta}_{t+1|t}$ and $\Sigma_{t+1|t}$.

\begin{lemma}\label{lem:LK}
Ignoring additive constants,
\begin{align}
L_t(\alpha)
&=
\log\det\!\big(\Lambda_t(\alpha)\big)
+
\epsilon_t(\alpha)^\top
\Lambda_t(\alpha)^{-1}
\epsilon_t(\alpha). \nonumber 
\end{align}

For SISO systems, 
\begin{align}\label{eq:MLL}
L_t(\alpha)
=
\log \Lambda_t(\alpha)
+
\frac{\epsilon_t(\alpha)^2}{\Lambda_t(\alpha)}.
\end{align}
\end{lemma}

The function $L_t(\alpha)$ represents the \emph{per-sample} negative log-likelihood. Its gradient $\nabla L_t(\alpha)$ and expected Hessian
\begin{equation}\label{eq:Hessian}
H_{\alpha,t}
:=
\E\left[
\frac{\partial^2 L_t(\alpha)}{\partial\alpha^2}
\mid {\cal D}_t^-
\right]
\end{equation}
are computed at time $t$. The following lemma provides expressions for these quantities.
\begin{lemma}\label{lemm:gradients}
Let us define
\[
\bar \epsilon_t(\alpha)
:=
\Lambda_t^{-1}(\alpha)\,\epsilon_t(\alpha),
\]
\[
\Psi
:=
\Lambda_t^{-1}(\alpha)
-
\bar \epsilon_t(\alpha)\bar \epsilon_t(\alpha)^\top.
\]
Then
\[
\frac{\partial \Lambda_t}{\partial \alpha}
=
2\alpha C_t\left(\Sigma_{t-1|t-1}-Q\right)C_t^\top.
\]
The gradient w.r.t. $\alpha$ of the negative log-likelihood of the measurement $y_t$, conditioned on the data ${\cal D}_t^-$, is
\begin{align*}
\nabla L_t(\alpha)
=
\operatorname{Tr}\!\left[
\Psi \frac{\partial \Lambda_t}{\partial \alpha}
\right]
-
2 \hat \theta^\top_{t-1|t-1} C_t^\top \bar \epsilon_t(\alpha).
\end{align*}
The expected Hessian w.r.t. $\alpha$ is given by 
\begin{align*}
H_{\alpha,t} \! = \! \operatorname{Tr}\! \!\left[
\Lambda_t^{-1}\frac{\partial \Lambda_t}{\partial \alpha}\Lambda_t^{-1}\frac{\partial \Lambda_t}{\partial \alpha}
\right] \! \! + \! 2\hat\theta_{t-1|t-1}^\top C_t^\top \! \Lambda_t^{-1} C_t\hat\theta_{t-1|t-1}.
\end{align*}
For SISO systems, $\Lambda_t$, $\epsilon_t$, and $\bar \epsilon_t$ are scalars, so
\begin{align*}
\nabla L_t(\alpha)
&=
\frac{\partial \Lambda_t}{\partial \alpha}\,
\frac{\Lambda_t-\epsilon_t^2}{\Lambda_t^2}
+
\frac{2\epsilon_t}{\Lambda_t}\frac{\partial \epsilon_t}{\partial \alpha},\;\;
\\
\text{ and } \quad H_{\alpha,t} &= \left(\frac{1}{\Lambda_t}\frac{\partial \Lambda_t}{\partial \alpha}\right)^2 + \frac{2}{\Lambda_t} \left(C_t \hat\theta_{t-1|t-1}\right)^2.
\end{align*}
\end{lemma}
In order to guarantee that $\alpha_t$ satisfies the constraints $\alpha_t \in [0,1]$ at all times, a re-parametrization such as the one described in Appendix \ref{sec:reparam} can be employed.

Ideally, one would update the hyperparameter $\alpha$ by minimizing the population form of \eqref{eq:MLL}, namely the expected one-step negative log-likelihood under the distribution of $y_t$. 
In the online setting, however, only a single realization $y_t$ is observed at each time $t$, and the hyperparameter update must be based on the corresponding single-sample loss, which is a noisy, Monte Carlo estimate of the population objective.
To address the noise that arises due to the single-sample, online context, we propose an exponentially-weighted gradient formulation for the hyperparameter tuning that exploits the following recursions

\begin{equation}\label{eq:Adapt:hyper}
\begin{array}{rcl}
\alpha_{t} & = & \alpha_{t-1} - \bar H_{t}^{-1} \bar \nabla L_t 
\end{array}
\end{equation}
where $\bar \nabla L_t$ and $\bar H_{t}$ are exponentially-weighted averages of $\nabla L_t(\alpha)$ and $H_{\alpha,t}$ with forgetting-factor $\beta$:
\begin{equation}\label{eq:avg:grad}
\begin{array}{rcl}
\bar \nabla L_t & = & (1-\beta)\bar \nabla L_{t-1} + \beta  \nabla L_t(\alpha)_{|\alpha = \alpha_{t-1}}, \\
\bar H_{t} & = & (1-\beta)\bar H_{t-1} + \beta  \left[H_{\alpha,t}\right]_{|\alpha = \alpha_{t-1}}. 
\end{array}
\end{equation}
A Newton/Hessian-based update is used for $\alpha$, as $\alpha$ must adapt quickly in rapidly-time-varying applications.
If desired, $\beta$ and $\sigma^2$ could be adapted based on data in a similar manner to $\alpha$, possibly exploiting time-scale  separation as in \cite{Borkar}. 

\begin{lemma}\label{lem:expFilter}
For some parameter vector $\nu \in \mathbb{R}^d$, e.g., $\nu=(\alpha)$, let $\{\nu_t\}_{t\ge 1}$ denote the sequence of iterates, and let the exponentially-filtered gradient be defined by
\begin{align*}
\bar \nabla L_t
:=
(1-\beta)\bar \nabla L_{t-1}
+
\beta \nabla L_t(\nu_{t-1}), \quad \beta \in [0,1].
\end{align*}
Define the exponentially-weighted average of the  linearized surrogate 
$$
\nonumber
\begin{array}{rcl}
\bar L_t(\nu)
&:= &
\sum_{\tau=1}^t
\beta(1-\beta)^{t-\tau} \widetilde L_\tau(\nu) \\
\widetilde L_\tau(\nu) & :=  &
\Bigl[
L_\tau(\nu_{\tau-1})
+
\nabla L_\tau(\nu_{\tau-1})^\top(\nu-\nu_{\tau-1})
\Bigr].
\end{array}
$$
If $\bar \nabla L_0=0$, then, for all $t$,
$
\bar \nabla L_t= \nabla \bar L_t(\nu)_{|
\nu = \nu_{t-1}
}.
$
\end{lemma}
Lemma \ref{lem:expFilter} shows that the exponentially-weighted gradient \eqref{eq:avg:grad} is exactly the gradient of the exponentially-weighted linearized surrogate $\bar L_t(\nu)$. 
By averaging information from past samples, $\bar\nabla L_t$ mitigates the variability of the single-sample gradient $\nabla L_t$, at the cost of temporal smoothing. 

\section{Proposed Data-Driven Predictive Control Algorithm}

Our proposed algorithm, implemented at each time step, is
\begin{algorithm}[H]
\caption{Algorithm 1: Recursive, FCE-minimizing DPC}
\label{alg:recursiveFCE}
\footnotesize
\begin{algorithmic}[1]
\State Using $\hat{\theta}_{t|t-1}$ and $\Sigma_{t|t-1}$, determine $u_f^\star$ by solving \eqref{eq:FCEcostSISO}.
\State Apply the first entry of $u_f^\star$ as $u_t$ and measure $y_t$.
\State Update the hyperparameter $\alpha_t$ by taking the Newton step \eqref{eq:Adapt:hyper}.
\State Using $C_t$ and $y_t$, iterate the Kalman filter \eqref{eq:Kalman:update}. 
\State Set $t\leftarrow t+1$ and repeat.
\end{algorithmic}
\end{algorithm}
At each time step, Algorithm \ref{alg:recursiveFCE} determines the sequence of future input actions $u_f^\star$ by minimizing the FCE through \eqref{eq:FCEcostSISO}, where the FCE is the posterior expectation of the cost $\mathcal{J}(u_f,y_f)$ given the data $\mathcal{D}_t^-$.
The prior on the ARX coefficients is zero-mean with covariance $Q$, 
which can encode characteristics such as smoothness and/or stability.
The predicted mean $\hat{\theta}_{t|t-1}$ and covariance $\Sigma_{t|t-1}$ of the ARX coefficients, obtained from the previous posterior via the $\alpha_t$-dependent Kalman prediction step in \eqref{eq:Kalman:update}, parameterize the FCE.

\section{Results}

In order to validate the proposed algorithm, we consider a time-varying dynamical system
$$
y_t = \sum_{k=1}^{\infty} h_{t,k} u_{t-k} + g_{t,k} n_{t-k}, 
$$
where $u_t$ is the controlled input and $n_t$ is a zero-mean white noise sequence. 
The time-dependent impulse responses $\{h_{t,k}\}$ and $\{g_{t,k}\}$ are obtained by interpolating between two linear dynamical systems using a hold--ramp-up--hold--ramp-down--hold scheduling profile, as described in the simulation protocol above.
Notably, System 1 is low pass with a static gain of roughly $40dB$, while System 2 is resonant with a static gain of roughly $5dB$, implying a significant change in their dynamic behavior.

\begin{algorithm}[t] 
\caption{Protocol used to generate the time-varying test system}
\label{sim:protocol}
\footnotesize
\begin{algorithmic}[1]

\State Form endpoint poles/zeros:
\Statex 
\(N_1 = [-0.1+0.1i, -0.1-0.1i],\quad N_2 = [0.1+0.95i, 0.1-0.95i]\)
\Statex 
\(D_1 = [0.9+0.01i, 0.9-0.01i],\quad D_2 = [0.6+0.6i, 0.6-0.6i]\)

\State Construct the scalar schedule $s_t$ for $t=1,\ldots,1000$:
\Statex \hspace{\algorithmicindent}
\(s_t = 0\) for samples \(1\) to \(100\)
\Statex \hspace{\algorithmicindent}
\(s_t\) ramps linearly from \(0\) to \(1\) over samples \(101\) to \(400\)
\Statex \hspace{\algorithmicindent}
\(s_t = 1\) for samples \(401\) to \(500\)
\Statex \hspace{\algorithmicindent}
\(s_t\) ramps linearly from \(1\) to \(0\) over samples \(501\) to \(800\)
\Statex \hspace{\algorithmicindent}
\(s_t = 0\) for samples \(801\) to \(1000\)

\State For each \(t=1,\dots,1000\), set (zeros and poles respectively)
\Statex \hspace{\algorithmicindent}
\(N_H(t,:) = (1-s_t)N_1 + s_t N_2\)
\Statex \hspace{\algorithmicindent}
\(D_H(t,:) = (1-s_t)D_1 + s_t D_2\)
\Statex \hspace{\algorithmicindent}
\(N_G(t,:) = 2\)
\Statex \hspace{\algorithmicindent}
\(D_G(t,:) = D_H(t,:)\)

\end{algorithmic}
\end{algorithm}

The initial $\alpha_0=.9999$. 
The control is designed to track the piecewise-constant reference shown in Fig.~\ref{fig:tracking}.  
At each time step $t$, the applied control is
$$u_t = u_t^\star + v_t$$
where $u^\star_t$ is the first entry of the FCE-minimizing $u^\star_f$, and $v_t$ is zero-mean noise with variance $0.01$ added to promote persistent excitation. The excitation variance trades off parameter identifiability and control performance.
Using $200$ off-line data points, $\sigma^2$ is set to the sample variance of the prediction errors. 
The parameter $q$ in \eqref{eq:cost} is $10^4$.
In addition, the smoothing parameter for the gradients and Hessians (see \eqref{eq:avg:grad}) is $\beta = 1/10$. 
The truncation length in \eqref{eq:predictorTV_ARX} is  $\hat \rho = 10$ and the control horizon $T=10$.

Our formulation requires an externally-provided prior for the parameter $\theta$, which determines the variance $Q$ of the noise in \eqref{eq:ssprior}, as stated in Proposition \ref{prop:invarance}.
We choose to use a DC prior \cite{pillonetto2022regularized}, which favors impulse-responses that are stable and smooth.
The DC prior hyperparameters are chosen by generating $100$ different, uniformly-sampled hyperparameter sets and then averaging their corresponding variances.


Figs.~\ref{fig:tracking} and \ref{fig:tracking2} report the tracking performance, demonstrating the importance of using regularization $r_t(u_f)$ in \eqref{eq:FCEcostSISO}. 
When the regularization is turned off, instabilities arise. 
\begin{figure}[h]
    \centering
    \includegraphics[width=0.99\columnwidth]{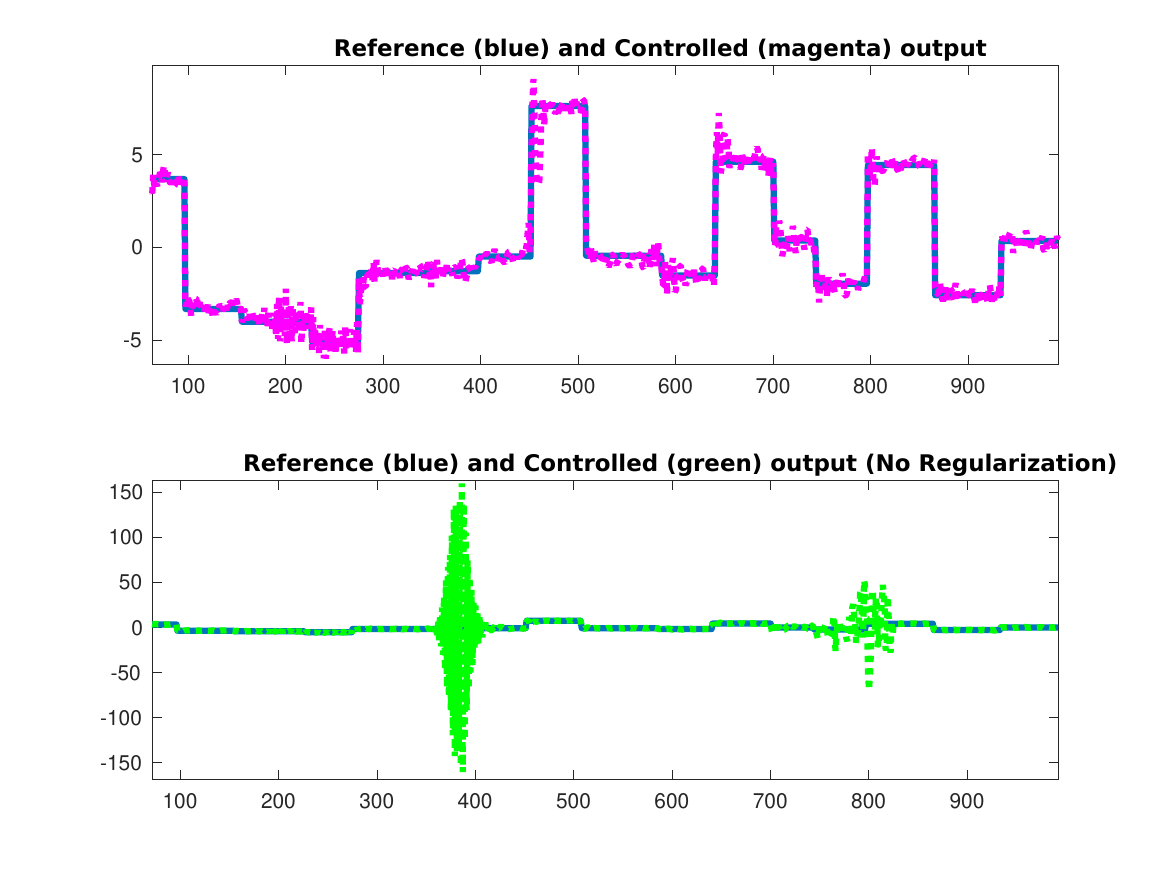}
    \caption{Reference (solid) vs controlled output (dotted): Top  with regularization, bottom  without  regularization $r_t(u_f)$.  
    }
    \label{fig:tracking}
\end{figure}
\begin{figure}[h]
    \centering
    \includegraphics[width=0.99\columnwidth]{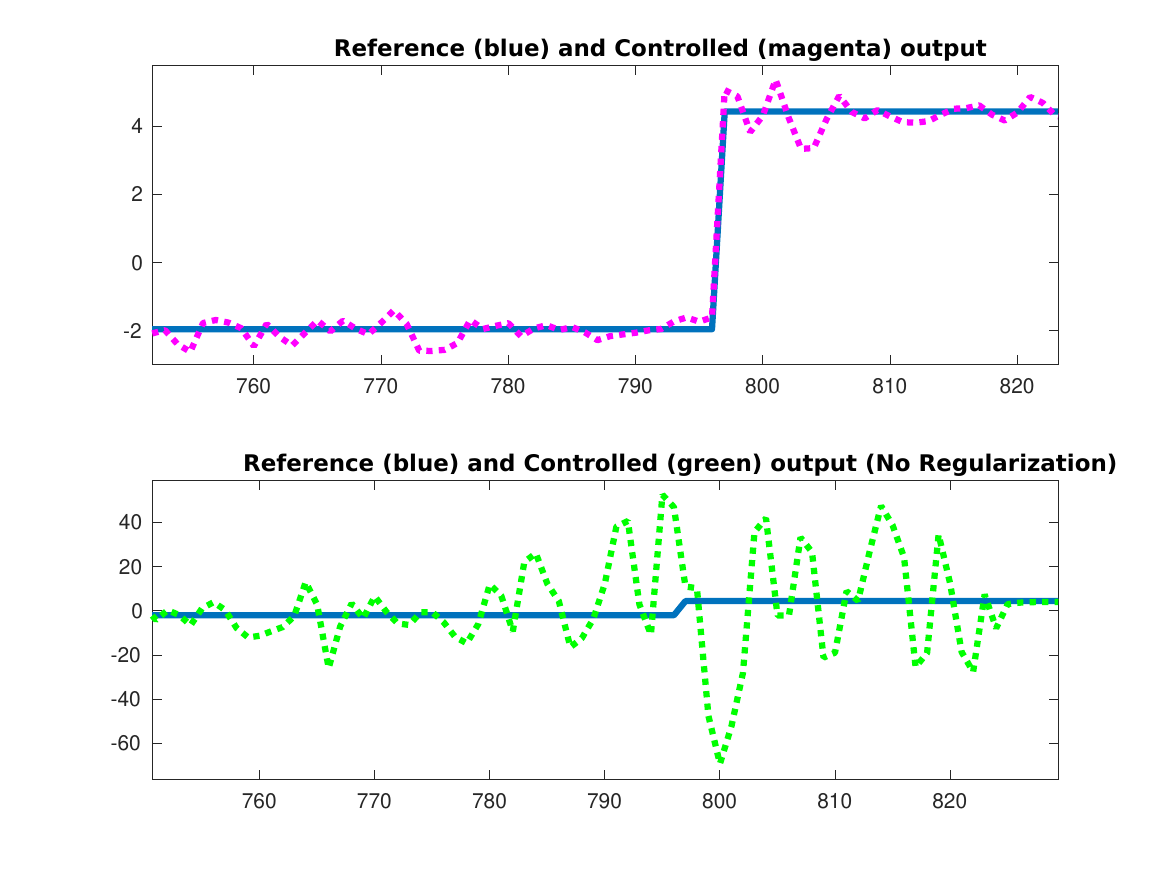}
    \caption{A zoom-in of Fig. \ref{fig:tracking} around time 800}
    \label{fig:tracking2}
\end{figure}

Figure \ref{fig:boxplots} shows running averages (exponentially weighted average) of the cost 
$\mathcal{J}(u_f,y_f)$ in \eqref{eq:cost}. 
On the left the FCE is optimized, whereas on the right the certainty equivalent objective is optimized. 
Once again, these results show the importance of regularization. 

\begin{figure}[h]
    \centering
\includegraphics[width=0.99\columnwidth]{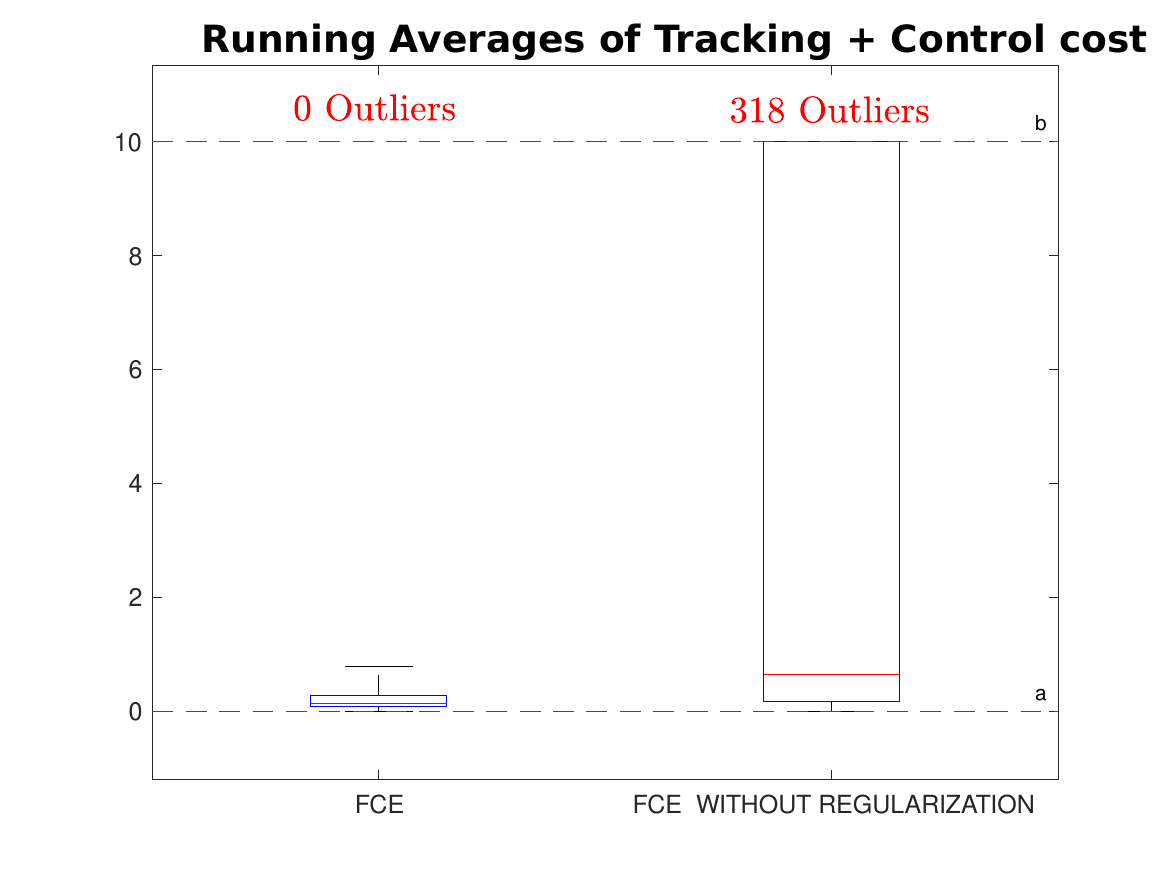}
    \caption{Boxplots of the running average of the cost $\mathcal{J}(u_f,y_f) = \|y_f - y_r\|^2 +  \frac{1}{q}\| u_f\|^2$ when the FCE-minimizing controller is run with regularization (left) and without (right).} 
    \label{fig:boxplots}
\end{figure}

\section{Conclusion}
This paper investigates an invariant, online Data-Driven Predictive Control algorithm that simultaneously estimates the ARX parameters that model the system, the ARX parameter uncertainty, and the process equation to propagate these parameter beliefs forward in time. 
The parameter-propagating process equation is designed to be invariant, maintaining the system parameter priors indefinitely. 
Both the ARX parameter and its uncertainty are used to determine the control action, which incorporates the parameter uncertainty as a regularization term.
Simulation results are promising and warrant further investigation. 
Future research directions will investigate directional forgetting and alternative, invariant process equations. 

\appendices

\section{intuition for the process equation adaptation}


Changing \(\alpha\) changes both the predicted mean and the predicted uncertainty of $\theta_t$. Thus, the $\alpha$-gradient has mean-fit and variance-fit components. 
For the SISO case,
\begin{align*}
\frac{\partial L_t}{\partial \alpha}
&=
\underbrace{
\frac{\partial \Lambda_t}{\partial \alpha}\,
\frac{\Lambda_t-\epsilon_t^2}{\Lambda_t^2}
}_{\text{variance-fit}}
+
\underbrace{
\frac{2\epsilon_t}{\Lambda_t}\frac{\partial \epsilon_t}{\partial \alpha}
}_{\text{mean-fit}}
\end{align*}
The mean-fit term reflects how $\alpha$ changes the predicted output itself. Because $\epsilon_t$ is the difference between the observed output and the one-step prediction, this term drives $\alpha$ in the direction that reduces that prediction error.

If, for example,  $\epsilon_t > 0$, then the prediction was too low. If, in addition, $\partial \epsilon_t / \partial \alpha > 0$, then increasing $\alpha$ makes $\epsilon_t$ larger. Since the mean-fit term is then positive, gradient descent decreases $\alpha$, moving in the direction that reduces $\epsilon_t$.

The variance-fit term adjusts \(\alpha\) so that the one-step predictive variance, computed from the Kalman filter covariance under \eqref{eq:ssprior}, better matches the observed innovation magnitude.

If, for example, \(\epsilon_t^2 > \Lambda_t\), then the realized innovation was larger than predicted. Furthermore, if, as is often the case, 
$
\frac{\partial \Lambda_t}{\partial \alpha}<0,
$
then the variance-fit term is positive and gradient-descent steps decrease \(\alpha\). 
This action is intuitive, as decreasing \(\alpha\) favors a model in which \(\theta_t\) varies more rapidly in time and, presuming \(\frac{\partial \Lambda_t}{\partial \alpha}<0\), also increases the one-step predictive variance \(\Lambda_t\).

\section{Proofs}

\subsection{Proof of Lemma \ref{lem:LK}:} 
Applying the Kalman filter to the gaussian $\theta_t$, we get
\begin{align*}
&p_{\alpha}\!\left(
y_t \mid {\cal D}_t^-,\{\alpha_\tau\}_{\tau=1,\ldots,t-1}
\right) \\
& \propto
\det\!\big(\Lambda_t(\alpha)\big)^{-1/2}
\exp\!\left(
-\frac12 \epsilon_t(\alpha)^\top
\Lambda_t(\alpha)^{-1}
\epsilon_t(\alpha)
\right).
\end{align*}
The general result follows from taking \(-2\log(\cdot)\) of this likelihood and dropping additive constants. 

\noindent For SISO, \(\Lambda_t(\alpha)\) and $\epsilon_t(\alpha)$ are scalar and
$
\det\!\big(\Lambda_t(\alpha)\big)
=
\Lambda_t(\alpha).
$

\subsection*{Proof of Lemma \ref{lemm:gradients}:} 
From Lemma \ref{lem:LK},
\begin{align*}
L_t(\alpha)
=
\log\det \Lambda_t
+
\epsilon_t^\top \Lambda_t^{-1}\epsilon_t.
\end{align*}
Using
\begin{align*}
\frac{\partial}{\partial x}\log\det \Lambda_t
&=
\operatorname{Tr}\!\left(
\Lambda_t^{-1}\frac{\partial \Lambda_t}{\partial x}
\right),\\
\frac{\partial}{\partial x}\bigl(\epsilon_t^\top \Lambda_t^{-1}\epsilon_t\bigr)
&=
2\left(\frac{\partial \epsilon_t}{\partial x}\right)^\top \Lambda_t^{-1}\epsilon_t
-
\epsilon_t^\top \Lambda_t^{-1}
\frac{\partial \Lambda_t}{\partial x}
\Lambda_t^{-1}\epsilon_t,
\end{align*}
we get
\begin{align*}
\frac{\partial L_t}{\partial x}
&=
\operatorname{Tr}\!\left[
\left(\Lambda_t^{-1}-\bar \epsilon_t\bar \epsilon_t^\top\right)
\frac{\partial \Lambda_t}{\partial x}
\right]
+
2\left(\frac{\partial \epsilon_t}{\partial x}\right)^\top \bar \epsilon_t \\
&=
\operatorname{Tr}\!\left[
\Psi \frac{\partial \Lambda_t}{\partial x}
\right]
+
2\left(\frac{\partial \epsilon_t}{\partial x}\right)^\top \bar \epsilon_t.
\end{align*}
Setting \(x=\alpha\), and using
\begin{align*}
\frac{\partial \epsilon_t}{\partial \alpha}
&=
-\,C_t\hat\theta_{t-1|t-1},\\
\frac{\partial \Lambda_t}{\partial \alpha}
&=
2\alpha C_t(\Sigma_{t-1|t-1}-Q)C_t^\top,
\end{align*}
gives the gradient w.r.t. $\alpha$.

To derive the expression for the expected Hessian $H_{\alpha,t}$ we use the relations
$$
\begin{array}{rcl}
\E\left[\epsilon_t(\alpha)|{\cal D}_t^-\right] = 0 \\
\E\left[\epsilon_t(\alpha)\epsilon_t^\top(\alpha)|{\cal D}_t^-\right] = \Lambda_t(\alpha)
\end{array}
$$
The second relation also implies that 
$$
\E\left[\bar\epsilon_t(\alpha)\bar\epsilon_t^\top(\alpha)|{\cal D}_t^-\right] = \Lambda_t^{-1}(\alpha)
$$
which, in turn, implies that
$$
\E\left[\Psi|{\cal D}_t^-\right] = 0. $$
Taking the derivative of
$$
\nabla L_t(\alpha)
=
\operatorname{Tr}\!\left[
\Psi \frac{\partial \Lambda_t}{\partial \alpha}
\right]
-
2 \hat \theta^\top_{t-1|t-1} C_t^\top \bar \epsilon_t(\alpha),
$$ gives
$$\begin{array}{rcl}
\frac{\partial}{\partial \alpha} \nabla L_t(\alpha) & = & \operatorname{Tr}\left[\frac{\partial \Psi}{\partial \alpha}\frac{\partial \Lambda_t}{\partial \alpha}+ \Psi \frac{\partial^2 \Lambda_t}{\partial \alpha^2}\right] + \\ & & + 2 \hat \theta^\top_{t-1|t-1} C_t^\top \Lambda^{-1}_t C_t \hat \theta_{t-1|t-1} + \\ & & +2 \hat\theta^\top_{t-1|t-1}  C_t^\top \Lambda^{-1}_t \frac{\partial \Lambda_t}{\partial \alpha} \bar \epsilon_t(\alpha)
\end{array}
$$
Now observe that 
$$\begin{array}{rcl}
\E \left[\Psi \frac{\partial^2 \Lambda_t}{\partial \alpha^2} |{\cal D}_t^-\right] 
& = & \E \left[\Psi   |{\cal D}_t^-\right]\frac{\partial^2 \Lambda_t}{\partial \alpha^2} \\
& = & 0
\end{array}
$$ and also 
$$\begin{array}{l}
\E \left[\hat\theta^\top_{t-1|t-1}  C_t^\top \Lambda^{-1}_t \frac{\partial \Lambda_t}{\partial \alpha} \bar \epsilon_t(\alpha) |{\cal D}_t^-\right]  \\
 = \hat\theta^\top_{t-1|t-1}  C_t^\top \Lambda^{-1}_t \frac{\partial \Lambda_t}{\partial \alpha}  \E \left[\bar \epsilon_t(\alpha)   |{\cal D}_t^-\right]  \\
 =  0
\end{array}
$$
In both cases the first equality  uses the fact that 
$\frac{\partial^2 \Lambda_t}{\partial \alpha^2}$ and $\hat\theta^\top_{t-1|t-1}  C_t^\top \Lambda^{-1}_t \frac{\partial \Lambda_t}{\partial \alpha}$ are  measurable w.r.t. ${\cal D}^-_t$.
Similarly, 
$$\begin{array}{rcl}
\operatorname{Tr}\left[\frac{\partial \Psi}{\partial \alpha}\frac{\partial \Lambda_t}{\partial \alpha}\right] & = &  
-2\operatorname{Tr}\left[\Lambda_t^{-1} \frac{\partial \Lambda_t}{\partial \alpha}\Psi \frac{\partial \Lambda_t}{\partial \alpha}\right] \\
& & -2 \operatorname{Tr}\left[\Lambda_t^{-1} \frac{\partial \epsilon_t}{\partial \alpha} \epsilon_t^\top \Lambda_t^{-1}  \frac{\partial \Lambda_t}{\partial \alpha}\right] \\
& & + \operatorname{Tr}\left[\Lambda_t^{-1} \frac{\partial \Lambda_t}{\partial \alpha}  \Lambda_t^{-1}  \frac{\partial \Lambda_t}{\partial \alpha}\right]
\end{array}
$$
and computing  the conditional expected value the first and second terms on the right hand side vanish, thus leading to 
$$
\begin{array}{rcl}
H_{\alpha,t}& := &  \E\left[\frac{\partial}{\partial \alpha} \nabla L_t(\alpha)|{\cal D}_t^-\right] \\
& = & \operatorname{Tr}\left[\Lambda_t^{-1} \frac{\partial \Lambda_t}{\partial \alpha}  \Lambda_t^{-1}  \frac{\partial \Lambda_t}{\partial \alpha}\right] \\ & & + 2 \hat \theta^\top_{t-1|t-1} C_t^\top \Lambda^{-1}_t C_t \hat \theta_{t-1|t-1}
\end{array}
$$

For SISO systems,
\begin{align*}
\Psi
=
\Lambda_t^{-1}-\bar \epsilon_t^2
=
\frac{\Lambda_t-\epsilon_t^2}{\Lambda_t^2},
\end{align*}
and the stated scalar formulas follow with algebra.

\subsection{Proof of Lemma \ref{lem:expFilter}:} 
\noindent 
Taking the gradient of $\bar L_t(\nu)$ w.r.t. $\nu$ we obtain
$
 \nabla \bar L_t (\nu)
= 
\sum_{\tau=1}^t
\beta(1 - \beta)^{t-\tau}\nabla L_\tau(\nu_{\tau-1})
$
and the results follow by induction by setting $\bar \nabla L_0 =0$.

\noindent Taking the gradient of $\widetilde L_t(\nu)$ w.r.t. $\nu$, the result follows. 

\section{Re-parameterization of $\alpha$}\label{sec:reparam}

Since $\alpha$ is constrained, we apply the following re-parameterizations when taking gradient steps.
Rather than adapting $\alpha$ directly, we adapt $\eta$,   where
$\alpha = a + b \frac{1}{1+e^{-\eta}}$. Ensuring $\eta \in [-1,1]$,
avoids vanishing gradients when adjusting $\eta$ online. 
While the technically allowed range for $\alpha$ is $[0,1]$, for the slowly varying systems considered here we restrict it to $[0.9,1]$. Faster variations may require a smaller lower bound. We choose $a$ and $b$ such that
$
1 = a + b \frac{1}{1+e^{-1}} \text{, and } 
.9 = a + b\frac{1}{1+e^{1}}.
$
$\eta$ is updated by taking a Newton step and then projecting the result on the box constraint $[-1,1]$.

\bibliographystyle{IEEEtran}
\bibliography{bibliography}

\end{document}